\documentclass[runningheads]{llncs}
\usepackage[T1]{fontenc}
\usepackage{graphicx}
\usepackage{amssymb}
\usepackage{amsmath}
\usepackage{xcolor}
\usepackage{todonotes}
\usepackage[colorlinks=true,allcolors=blue]{hyperref}
\usepackage{xfrac}

\usepackage{algorithm}
\usepackage{algpseudocodex}

\usepackage{hyperref}
\usepackage{mathtools}

\begin{document}
\title{On the Time Complexity of Finding a Well-Spread Perfect Matching in Bridgeless Cubic Graphs%
  \thanks{Supported by the project GAUK182623 of the Charles University Grant Agency and by 
  grant 25-16627S of the Czech Science Foundation.
  }}
\titlerunning{Well-Spread Perfect Matching}
%

\author{
Babak Ghanbari \and
Robert Šámal}
\authorrunning{B. Ghanbari and R. Šámal}
%
\institute{Computer Science Institute, Charles University, Prague, Czech Republic \\
\email{ \{babak,samal\}@iuuk.mff.cuni.cz }\\
\url{} }
\maketitle              
\begin{abstract}
We present an algorithm for finding a perfect matching in a $3$-edge-connected cubic graph that intersects every $3$-edge cut in exactly one edge.
Specifically, we propose an algorithm with a time complexity of $O(n \log^4 n)$, which significantly improves upon the previously known $O(n^3)$-time algorithms
for the same problem. The technique we use for the improvement is efficient use of the cactus model of 3-edge cuts. 
As an application, we use our algorithm to compute embeddings of $3$-edge-connected cubic graphs with limited number of singular edges (i.e.,
edges that are twice in the boundary of one face) in  $O(n \log^4 n)$ time; this application contributes to the study of the well-known Cycle Double Cover conjecture. 
\end{abstract}
\section{Introduction}

It is well known that every bridgeless cubic graph admits a perfect matching~\cite{JuliusPetersen}. Using Edmonds' blossom algorithm, a
maximum matching can be found in polynomial time. Gabow~\cite{Gabow} demonstrated that the weighted matching problem on general graphs
can be solved in time $O(n(m + n \log n))$.

Consequently, the minimum weight perfect matching problem for bridgeless cubic graphs can be solved in $O(n^2 \log n)$ time. Diks and
Sta\'nczyk~\cite{Diks} proposed an improved algorithm for finding perfect matchings in bridgeless cubic graphs with time complexity of
$O(n \log^2 n)$.

In this paper we study the complexity of finding a perfect matching~$M$ in a bridgeless cubic graph such that in every 
3-edge cut $M$~contains exactly one edge. By parity, $M$ can contain one or three edges in a 3-edge cut; so our condition says 
no 3-edge cut is contained in $M$. We shortly express this by saying $M$ is \emph{well-spread}. 

Kaiser et al. \cite{Kaiser} showed that a well-spread perfect matching exists in every bridgeless cubic graph (and went on to 
use this to find how much of the graph can be covered by two, three, etc. perfect matchings). 
Their proof uses Edmonds' perfect matching polytope (vector $(1/3, 1/3, \dots)$ is a fractional perfect matching, thus 
it is a convex combination of perfect matchings -- each of them is well-spread). This, however, doesn't yield an efficient algorithm. 

In \cite{Boyd}, Boyd et al. focus on developing algorithms for finding $2$-factors in bridgeless cubic graphs that cover specific
edge-cuts, which brings these $2$-factors closer to Hamiltonian cycles. They provide an efficient algorithm that finds a minimum-weight $2$-factor
that covers all $3$-edge cuts in weighted bridgeless cubic graphs (in contrast with finding a Hamiltonian cycle, which is 
computationally hard). They also provide both a polyhedral description of such $2$-factors
and of their complements -- well-spread perfect matchings. This is the first known polynomial-time algorithm for this
problem, with a time complexity of $O(n^3)$, where $n$ is the number of vertices. In this work, we improve this result for
$3$-edge-connected cubic graphs using \hyperref[algorithm-3cut]{Algorithm~1} with time complexity $O(n \log^4 n)$. 

Boyd et al. find a ``peripheral 3-edge-cut'' (a cut such that one side of it is internally 4-edge-connected)
and then use recursion. We save computation by using the cactus model, also called a tree of cuts. 
In this tree it is easy to find the peripheral cut and also to update the tree when we contract a part of the cut. 
This leads to much improved time complexity, although only for 3-edge-connected graphs. 

We then proceed to apply this result to the study of the well-known Cycle Double Cover (or CDC) Conjecture
\cite{szekeres_1973,seymour1979sums}. In the language of graph embedding, the CDC conjecture is equivalent to every bridgeless cubic graph
having a surface embedding with no singular edges (i.e., with no edge that is on the boundary of one face twice). As an
application of our work, for a $3$-edge-connected cubic graph we can find embeddings with a bounded number of singular
edges in time $O(n \log^4 n)$. 

The structure of the paper is as follows. In the next section, we introduce the necessary definitions and concepts. In Section~3, we
discuss the cactus model, a key component in deriving our results. Section~4 presents our main result, 
an efficient algorithm to find the well-spread perfect matching. 
Finally, in Section~5, we apply our algorithm to get results about the well-known CDC conjecture.

\section{Preliminaries}
Let $G = (V, E)$ be a graph with vertex set $V$ and edge set $E$. A graph in which all vertices have degree $3$ is called \emph{cubic}.
A \emph{cycle} is a connected $2$-regular graph. A \emph{bridge} in a graph $G$ is an edge whose removal increases the number of
components of $G$. Equivalently, a bridge is an edge that is not contained in any cycles of $G$. A graph is \emph{bridgeless} if it
contains no bridge. An \emph{edge cut} in a graph is a set of edges whose removal increases the number of connected components of the
graph. An edge cut $C$ in $G$ is called a \emph{non-trivial cut} if every component of $G-C$ has at least two vertices. Otherwise it is called \emph{trivial}. A \emph{$k$-edge-cut} in a graph is an edge cut that contains exactly $k$ edges. A connected graph is \emph{$k$-edge-connected}
if it remains connected whenever fewer than $k$ edges are removed. 
A graph is \emph{cyclically $k$-edge-connected}, if at least $k$ edges must be removed to disconnect 
it into two components such that each component contains a cycle. We say that a subset $F\subset E$ \emph{covers} an edge-cut~$D$ 
if~$F \cap D \neq \emptyset$. For a subset $S \subseteq E$, $G/S$ is the graph obtained from $G$ by contracting all the edges in $S$.
Note that we keep multiple edges and only remove loops in such contraction. 

Petersen’s theorem~\cite{JuliusPetersen} states that every bridgeless cubic graph contains a perfect matching.
Sch\"{o}nberger~\cite{Schonberger1934} proved the following strengthened form of Petersen’s theorem.

\begin{theorem}[\cite{Schonberger1934}]\label{Schonberger}
    Let $G = (V, E)$ be a bridgeless cubic graph with specified edge $e^* \in E$. Then there exists a perfect matching of G that
    contains $e^*$.
\end{theorem} 
A perfect matching containing a specified edge $e^*$ in a bridgeless cubic graph can be found in $O(n \log^4 n)$ time~\cite{Therese}.

A graph $G$ is \emph{embedded} in a surface $S$ if the vertices of $G$ are distinct elements of $S$ and every edge of $G$ is a simple
arc connecting in $S$ the two vertices which it joins in $G$, such that its interior is disjoint from other edges and vertices. An
\emph{embedding} of a graph $G$ in $S$ is an isomorphism of $G$ with a graph $G'$ embedded in $S$. 

Let $G$ be a graph that is cellularly embedded in a surface $S$, that is, every face is homeomorphic to an open disk. Let $\pi = \{\pi_v
| v \in V(G)\}$ where $\pi_v$ is the cyclic permutation of the edges incident with the vertex $v$ such that $\pi_v(e)$ is the successor
of $e$ in the clockwise ordering around $v$. The cyclic permutation $\pi_v$ is called the \emph{local rotation} at $v$, and the set
$\pi$ is the \emph{rotation system} of the given embedding of $G$ in $S$. 

Let $G$ be a connected multigraph. A \emph{combinatorial embedding} of $G$ is a
pair $(\pi, \lambda)$ where $\pi = \{\pi_v | v \in V(G)\}$ is a rotation system,
and $\lambda$ is a signature mapping which assigns to each edge $e \in E(G)$ a
sign $\lambda(e) \in \{-1, 1\}$. If $e$ is an edge incident with $v \in V(G)$,
then the cyclic sequence $e, \pi_v(e), \pi_v^2(e), \dots$ is called the
$\pi$-\emph{clockwise ordering} around $v$ (or the \emph{local rotation} at
$v$). Given an embedding $(\pi, \lambda)$ of $G$ we say that $G$ is 
$(\pi, \lambda)$-\emph{embedded}. It is known that the combinatorial embedding uniquely determines a cellular embedding to some surface,
up to homeomorphism~{\cite[Theorems 3.2.4 and 3.3.1]{Mohar}}.

A \emph{closed walk} is a sequence $(v_0, e_0, v_1,$ $e_1, \dots, e_{n-1}, v_n)$ where $v_0 = v_n$ and 
for every $i$, $e_i = \{v_i,v_{i+1}\}$. We say that a collection of closed walks $C_1, \dots, C_n$ forms a 
\emph{partial circuit double cover} (or \emph{partial CDC}) if each edge is covered at most once by one of $C_i$'s or exactly 
twice by two different $C_i$ and $C_j$, and for every vertex $v$ and edges $e$ and $f$ where $e \cap f = \{v\}$, there exists at most
one closed walk $C_i$ such that $\{e, f\} \subseteq E(C_i)$.

\section{The Cactus Model}
Here we follow the notations and definition of \cite{Nagamochi-Book}. A general way to represent a subset of cuts within a graph $G$ involves constructing a cactus 
representation~$(T, \varphi)$. In this representation, $T$ is a graph, and $\varphi$ is a function mapping vertices from $V(G)$ to
$V(T)$. The cactus representation (or the \emph{cactus model}) was introduced in \cite{Dinic1976}. It was further utilized in
\cite{Gabow1991ApplicationsOA,Dinitz1993}. This model plays a crucial role in understanding both edge connectivity and
graph rigidity problems.

\begin{definition}
Let $G$ be a graph. A pair $(T, \varphi)$ with $\varphi: V(G) \to V(T)$ is called a cactus representation (or cactus model) for the graph $G$ if it meets the following criteria:
    \begin{enumerate}
        \item For an arbitrary minimum cut $\{S, V(T) - S\} \in C(T)$, the cut 
        $\{X, \bar{X}\}$ in $G$ defined by 
        $X = \{ u \in V(G) \mid \varphi(u) \in S \}$, and $\bar{X} = \{ u \in V(G) \mid \varphi(u) \in V(T) - S \}$
        is a minimum cut in $G$.
    
        \item Conversely, for every minimum cut $\{X, \bar{X}\}$ in $G$, there exists a minimum cut $\{S, V(T) - S\} \in C(T)$ such that
        $X = \{ u \in V(G) \mid \varphi(u) \in S \}$, and $\bar{X} = \{ u \in V(G) \mid \varphi(u) \in V(T) - S \}.$
    \end{enumerate}
\end{definition}

In this context, we use \emph{vertex} when referring to elements of $V(G)$ and \emph{node} for elements of $V(T)$. The set
$V(T)$ may include a node $x$ that does not correspond to any vertex $v \in V(G)$ with $\varphi(v) = x$; such a node is referred to as
an \emph{empty node}. We use $C(T)$ for the set of all minimum cuts in~$T$.

It was shown in \cite{Dinic1976} that every $G$ has a cactus representation using a \emph{cactus graph} -- a multigraph where 
every edge is in exactly one cycle. A special case of such graph is a tree with every edge doubled. 
It is known that when the connectivity is odd, the cactus representation is of this special type. In this case we will just call it a
tree of cuts (or cactus tree). For a $3$-edge-connected cubic graph $G$, the vertices of the graph $G$ correspond to the leaf nodes of
the cactus tree $T$ and all the interior nodes of $T$ are empty nodes (see Fig.~\ref{cactus-model}).  
We will also use the following estimate for the size of~$T$. 

\begin{figure}
    \begin{center}
        \includegraphics[scale=.18]{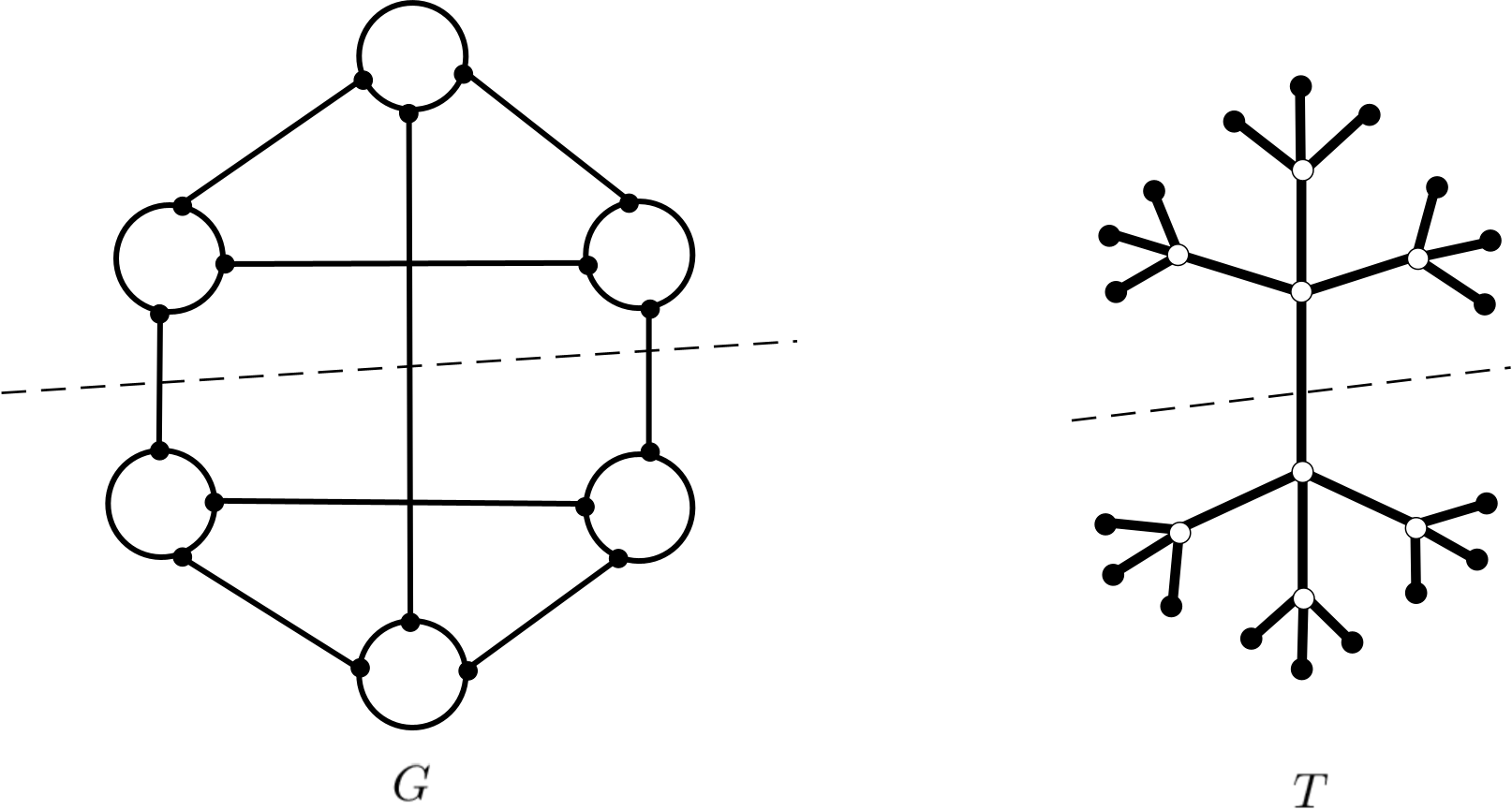}
    \end{center}
    \caption{A 3-edge-connected cubic graph $G$ and its cactus tree $T$. Small white circles represent empty vertices. Dashed lines represent a non-trivial $3$-edge cut in $G$ and its equivalent edge in $T$.}
    \label{cactus-model}
\end{figure}

\begin{lemma}[\cite{Dinitz1993}]\label{lemma-cactus}
    Let $G$ be a $3$-edge-connected cubic graph with $n$~vertices and let $T$ be its cactus tree. Then 
    $$|E(T)| \leq 2n-3.$$   
\end{lemma}

\section{Well-Spread Perfect Matchings in Bridgeless Cubic Graphs}\label{3cut-alg-section}

\def\Abar{{\bar A}} 
\def\abar{{\bar a}} 

In this section, we present an algorithm that finds a perfect matching in a $3$-edge-connected cubic graph that intersects all $3$-edge cuts
in exactly one edge and is substantially faster than the algorithm introduced by Boyd et al.~\cite{Boyd}. 
We first explain the underlying technique. Let $G$ be a graph with a 3-edge cut $\{A, \Abar\}$; that is
$|\delta(A)| = 3$. We let $G_1 = G/A$ (all vertices in~$A$ are contracted to a single vertex~$a$)
and $G_2 = G/\Abar$ (with a new vertex~$\abar$). A well-spread matching $M$ in~$G$ contains exactly one edge of~$\delta(A)$, 
thus it gives us a perfect matching~$M_1$ in~$G_1$ and $M_2$ in~$G_2$. As edges out of $a$ in~$G_1$ 
and out of $\abar$ in~$G_2$ are in natural correspondence with edges of~$\delta(A)$, we will identify them 
and say that $M_1$ and $M_2$ \emph{agree on $\delta(A)$}. 
For our algorithm the following is crucial.

\begin{theorem}\label{thm:PMcombine}
  Let $G$, $A$, $G_1$ and $G_2$ be as above. 
  Given a well-spread perfect matching $M_1$ in~$G_1$ and $M_2$ in~$G_2$ that agree on the cut $\delta(A)$, 
  the union $M_1 \cup M_2$ is a well-spread perfect matching in~$G$. 
\end{theorem}

\begin{proof}
  The fact that $M_1 \cup M_2$ is a perfect matching is clear as these matchings agree on~$\delta(A)$. 
  The fact that $M$ is well-spread follows from the fact that 3-edge cuts in a 3-edge-connected graph 
  do not cross, see \cite{Boyd} for details. 
  \qed
\end{proof}

Boyd et al.~\cite{Boyd} use the above theorem together with repeated finding of a ``peripheral 3-edge-cut'' $\{A, \Abar\}$: 
one where $G/A$ is internally 4-edge-connected. We refine their approach by utilizing (and updating) the model for 
all 3-edge cuts that will allow us to find peripheral cuts quickly. Also, we will use the model of the cuts 
to quickly decompose the graph $G$ in $G/A$ and $G/\Abar$ and then combine them back. 

We will also use the following information about structure of cuts in $G_1$ and $G_2$.

\begin{theorem}\label{thm:cactuscombine}
  Let $G$, $A$, $G_1$, $G_2$, $A$ and $\Abar$ be as above. 
  Let $(\varphi, T)$ be the cactus tree for~$G$. 
  Let $e = x_1 x_2$ be the edge of~$T$ corresponding to~$A$; that is the components of 
  $T-e$ are $T'_1$, $T'_2$, $x_i \in V(T'_i)$ and $A = \varphi^{-1}(V(T'_2))$. 
  Put $T_i = T'_i + e$ for $i=1, 2$. 
  Define $\varphi_1$ as a restriction of $\varphi$ to $\Abar$ 
  and put $\varphi_1(a) = x_1$. We define $\varphi_2$ symmetrically. 

  Then $(\varphi_i, T_i)$ is the cactus tree representation of 3-edge cuts in~$G_i$
  (for $i=1,2$). 
\end{theorem}

The proof follows the same idea as that of Theorem~\ref{thm:PMcombine}: minimum odd edge-cuts do not cross. 
We omit the details in this extended abstract. 


In our application we have a 3-regular 3-edge-connected graph, so all vertices 
are mapped to leaf nodes of the cactus tree and nontrivial 3-cuts correspond to edges that are not incident with a leaf. 
In particular, when the cactus tree is a star there are no nontrivial 3-cuts and we may simply find any perfect matching. 
In a sense, the purpose of our algorithm is to efficiently combine matchings coming from these ``special cases''
of internally 4-edge-connected graphs. 

To this end, we pick a root in the cut model~$T$. For a vertex $x$, let $e_x$ be the edge connecting $x$ to its parent. 
We let $C_x$ be the three edges in the cut corresponding to~$e_x$. We let $A_x$ be the vertices 
of~$G$ that are mapped by $\varphi$ to descendants of~$x$ (including~$x$), so that $C_x = \delta(A_x)$. 

We do not use this definition directly, instead we recursively update $A_x$, $C_x$ while traversing the tree. 
We can compute $A_x$ as a union of $A_y$ for all children $y$ of~$x$. 
We can compute $C_x$ by going over $C_y$ for all children $y$ of~$x$, excluding 
edges that lie within $A_x$ (between $A_y$ and~$A_{y'}$ for two children of~$x$). 
To do this efficiently we utilize the Union-Find algorithm. 

In a typical step of the decomposition part of the algorithm, we use Theorem~\ref{thm:cactuscombine}, 
with the new graphs being $G_x$ (which we store for later) and updated version of~$G$. 
To illustrate the algorithm, in Figure~\ref{cactus-updated} we show one step including the 
updated trees. We see that $T_x$ is a star -- which means $G_x$ is internally 4-edge-connected. 
Updated $T$ simply removes children of~$x$. However, we actually do not need to update~$T$ and 
compute $T_x$ in the algorithm. 

In a typical step of the assembling part of the algorithm, we use Theorem~\ref{thm:PMcombine}. 
For graph~$G$ we have our well-spread matching from the recursion. For the other graph, 
$G_x$, we compute it easily, as the graph is internally 4-edge-connected, so any perfect matching 
will do. 
Figure~\ref{Tree-Process2} illustrates this part of the algorithm. 

\begin{theorem}
    The output of \hyperref[algorithm-3cut]{Algorithm~1} is a well-spread perfect matching.
\end{theorem}
 
\begin{proof}
  Follows from Theorem~\ref{thm:PMcombine} and~\ref{thm:cactuscombine} and the discussion above. 
  \qed
\end{proof}

\begin{algorithm}
\caption{}\label{algorithm-3cut}
\begin{algorithmic}[0]
  \Require A $3$-edge-connected cubic graph $G$.
  \Ensure A well-spread perfect matching $M$.
    \State Find $\varphi: G \to T$ using \cite{Gabow1991ApplicationsOA}.
    \State Choose any node of $T$ as root.
    \LComment{Decomposing $G$ along $3$-edge-cuts}  
    \ForAll{$x \in V(T)$ in post-order}
        \If {$x$ is the root}  \Comment{no proper $3$-edge-cut remains}
          \State Break out of the for loop. 
        \ElsIf {$x$ is a leaf}                \Comment{Dealing with a trivial cut}
          \State $A_x \gets \varphi^{-1}(x)$  \Comment{Single vertex}
          \State $C_x \gets \delta(A_x)$      \Comment{three edges incident to it}
        \Else 
            \State $A_x \gets \bigcup \{ A_y \mid \mbox{$y$ a child of~$x$} \}$   \Comment{Use Union-Find}
            \State $C_x \gets \{u,v\} \in \bigcup \{ C_y \mid \mbox{$y$ a child of~$x$\} } \mbox{s.t. 
               only one of $u$, $v$ is in $A_x$}$ 
            \State $G_x \gets  G/\Abar_x$   \Comment{contract $\Abar_x$ to a single vertex $\abar_x$} 
            \State put a link to $C_x$, $G_x$ to $x$
            \State $G \gets  G/A_x$  \Comment{contract $A_x$ to a single vertex $a_x$} 
            \State $\varphi \gets \varphi/\Abar_x \cup \{(a_x, x)\}$
          \EndIf
          
    \EndFor
    \LComment{Assembling the perfect matching} 
    \ForAll{$x \in V(T)$ in reverse pre-order}
       \If{$x$ is a root of $T$}
          \State Find any perfect matching $M$ in $G$.   \Comment{$G$ is now internally 4-edge-connected}
       \ElsIf{$x$ is not a leaf}
         \LComment{First we find which edge of~$C_x$ does $M$ use.} 
         \State $b \gets \text{vertex }~s.t.~a_xb \in M$
         \State $a  \gets \text{vertex }~s.t.~ab \in C_x$
         \State Find a perfect matching $M_x$ in $G_x$ using the edge $a \abar_x$.  \Comment{$G_x$ is int.4-edge-connected}
         \State $M \gets M \cup M_x$
       \EndIf
    \EndFor
    \State \Return $M$
\end{algorithmic}
\end{algorithm}

\begin{figure}
    \begin{center}
        \includegraphics[scale=.18]{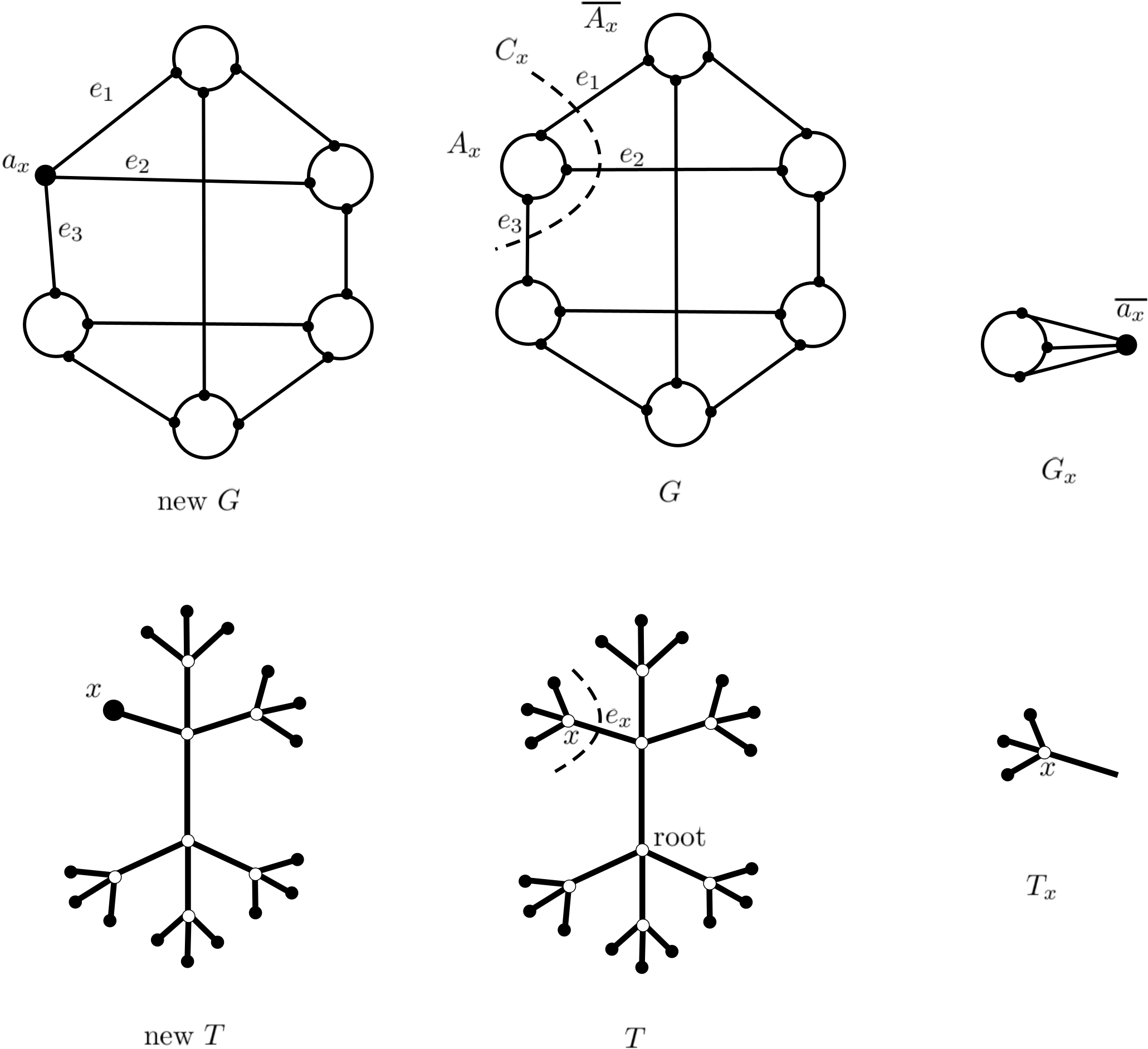}
    \end{center}
    \caption{A typical step in the decomposition part of the algorithm. We show a concrete 
      3-edge-connected cubic graph $G$ and its cactus tree $T$, together with the updated version 
      of $G$ and $T$ and the internally 4-edge-connected graph $G_x$ with the star cactus tree $T_x$. 
      }
    \label{cactus-updated}
\end{figure}

\begin{figure}
    \begin{center}
        \includegraphics[scale=.18]{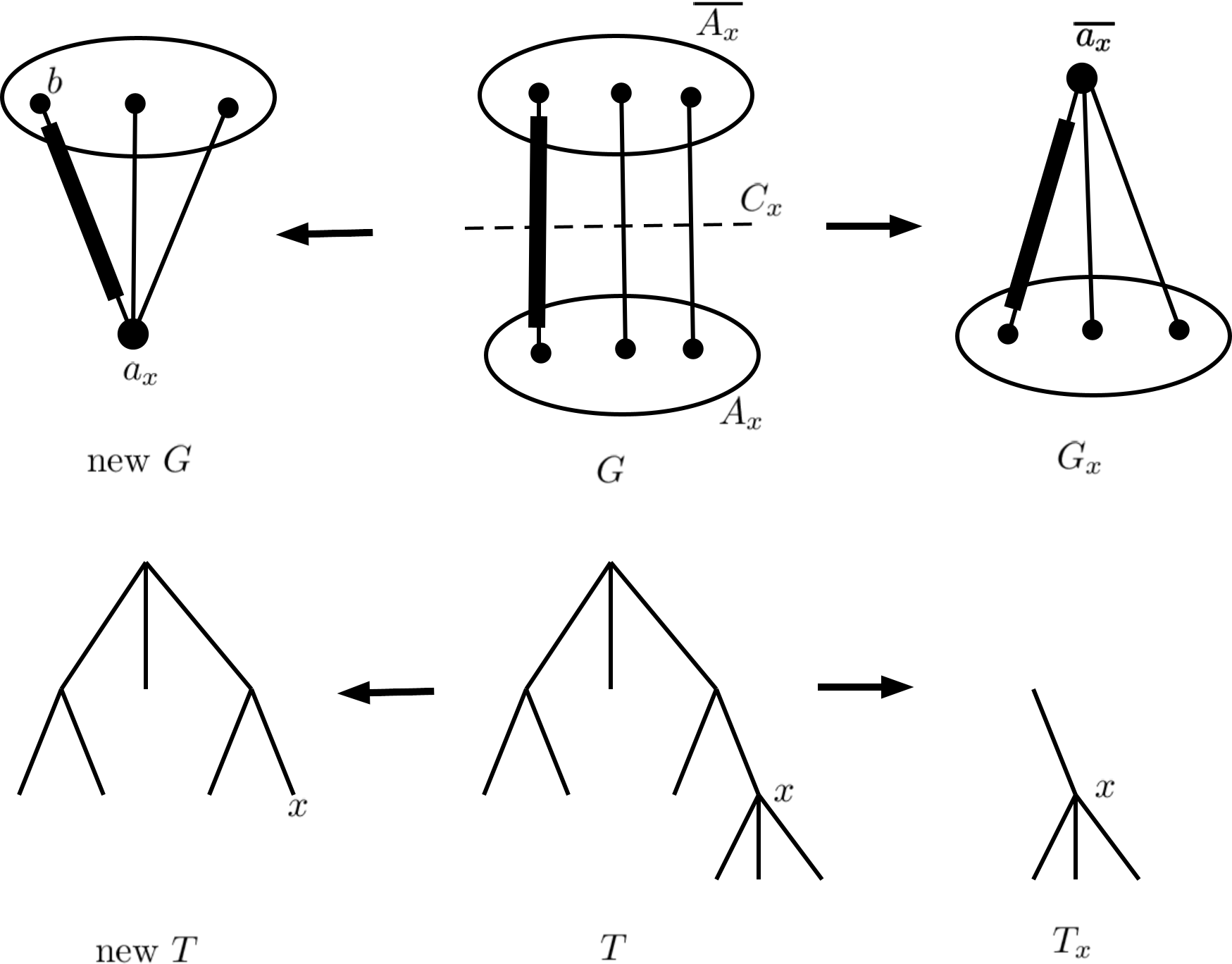}
    \end{center}
    \caption{A typical step in the assembly part of the algorithm: we combine 
      a perfect matching $M$ in ``new G'' and $M_x$ in $G_x$ to get a matching 
      in the original graph~$G$.}
    \label{Tree-Process2}
\end{figure}

\subsection{Complexity Analysis}
The construction of the tree of cuts (the cactus model) can be done in time $O(n \log n)$~\cite{Gabow1991ApplicationsOA}.
By Lemma~\ref{lemma-cactus} we have $|V(T)| = O(n)$. 
Next, the algorithm decomposes the graph $G$ along $3$-edge cuts. Here the main work is keeping track of the 
sets $A_x$ and $C_x$, updating the graph~$G$ and $G_x$. For this we use the standard Union-Find algorithm 
plus a constant amount of work for changing three edges of the cut $C_x$ at each step. \
Thus the total time required for this loop is $O(n \alpha(n))$. 

In the  assembly step, we repeatedly look for a perfect matching in an internally 4-edge-connected cubic graph~$G_x$ 
containing a given edge. By \cite{Therese}, this can be done in time $O(|G_x|\log^4|G_x|)$.
(The initial step where $x$ is the root is faster, as we don't need to specify an edge, but this bound works as well.) 
The remaining work (combining the matchings) needs only a constant time per edge, so linear in total. 

Note that $|G_x|$ is the degree of~$x$ in~$T$; let $d_x$ denote this quantity. 
We know that $\sum_x d_x = 2|E(T)| = O(n)$. The time complexity is thus 
$$
   O(n \log n) + O(n \alpha (n)) + O(n) + \sum_{x \in V(T)}  d_x \log^4 d_x . 
$$

So the leading term in the time complexity is the final sum. 
As the function $f(d) = d \log^4 d$ is convex, the sum will be maximal when one $d_x$ is as large as possible (namely, $O(n)$)
and the other as small as possible (namely, 1) -- otherwise we can increase the sum of $f(d_x)$, while keeping sum of $d_x$ constant. 
For this extreme case we get our bound for time complexity of the algorithm, 
$O(n \log^4 n)$.

\begin{theorem}
    \hyperref[algorithm-3cut]{Algorithm~1} finds a well-spread perfect matching $M$ in a $3$-edge-connected cubic graph $G$ in time $O(n \log^4 n)$.
\end{theorem}

\section{Application}
In~\cite{Ghanbari}, Ghanbari and \v{S}amal establish an upper bound of $\frac{n}{10}$ on the number of singular edges in an embedding of a bridgeless cubic graph on a surface. They also raise the question of how efficiently one can find a perfect matching in a bridgeless cubic graph that contains no odd cut of size $3$ — an essential step in determining the time complexity of finding such an embedding. Here, we answer their question for $3$-edge connected cubic graphs and describe an algorithm that constructs an embedding with at most $\frac{n}{10}$ bad edges. To provide context, we first state their results.

\begin{lemma}[\cite{Ghanbari}]\label{Lemma-PCDC}
Let $G$ be a bridgeless cubic graph, and $C_1, C_2, \dots, C_t$ be a collection of closed walks in $G$. If 
$C_1, C_2, \dots, C_t$ form a partial CDC, 
then there is an embedding $(\pi, \lambda)$ of $G$ where $C_1, \dots, C_t$ are some of the facial walks of $(\pi, \lambda)$.
Moreover, such an embedding can be found by a linear time algorithm. 
\end{lemma}

\begin{theorem}[\cite{Ghanbari}]\label{THM-(n/10)-bound}
Let $G$ be a bridgeless cubic graph. There exists an embedding of $G$ with at most $\frac{n}{10}$ singular edges.
\end{theorem}

To prove Theorem~\ref{THM-(n/10)-bound}, in~\cite{Ghanbari}, we first construct a perfect matching $M_1$ such that $M_1$ contains no odd cut of size 3. This guarantees the existence of another perfect matching  $M_2$  satisfying $|M_1 \cap M_2| \leq \frac{n}{10}$ (see~\cite{Kaiser}). Next, decompose $G\backslash M_1$ and $M_1 \Delta M_2$ into circuits $C_1, C_2, \dots C_t$. These circuits form a partial CDC and by Lemma~\ref{Lemma-PCDC} can be extended to an embedding in which the only potential singular edges are those that are not covered by $C_1, C_2, \dots C_t$, that is, edges in $M_1 \cap M_2$. See~\cite{Ghanbari} for the details. 

Suppose that the cubic graph $G$ is $3$-edge-connected. To find $M_1$, we employ \hyperref[algorithm-3cut]{Algorithm~1}, which finds a perfect matching in time $O(n \log^4 n)$. Consequently, determining the overall time complexity reduces to finding the perfect matching $M_2$. We utilize Diks and Stańczyk's algorithm~\cite{Diks} to find $M_2$ in time $O(n \log^2 n)$.  Thus, the embedding guaranteed by Theorem~\ref{THM-(n/10)-bound} can be found in total time  
\[
O(n \log^2 n) + O(n \log^4 n) = O(n \log^4 n).
\]



\section{Conclusion}
We studied a new approach to finding a perfect matching that includes all $3$-edge cuts in a $3$-edge-connected cubic graph, utilizing
the cactus representation of the graph. We referred to such a perfect matching as \emph{well-spread}. In Section~\ref{3cut-alg-section},
we presented an algorithm that finds a well-spread perfect matching in a $3$-edge-connected cubic graph in time $O(n\log^4 n)$. 

In general, the best-known algorithm for finding a well-spread perfect matching in a bridgeless cubic graph has a time complexity of
$O(n^3)$~\cite{Boyd}. In the future, it would be interesting to determine whether the cactus representation can be leveraged to develop
a faster algorithm for finding such a well-spread perfect matching.
However, the issue is that representing 3-edge-cuts in a graph that may contain 2-edge-cuts is more complicated.


\bibliographystyle{alpha}
\bibliography{RN.bib}
\end{document}